\documentclass[10pt]{amsart}
\usepackage{amssymb}
\usepackage{amsmath,amssymb,amsfonts,amsthm,
latexsym, amscd, amsfonts, epsfig, eepic,epic}
\usepackage{mathrsfs}
\usepackage{color}
\usepackage{eucal}
\usepackage{eufrak}
\usepackage[all]{xypic}
\usepackage{xspace}

\theoremstyle{plain}
\newtheorem{theorem}{Theorem}

\newtheorem{lemma}{Lemma}

\theoremstyle{definition}

\theoremstyle{example}

\theoremstyle{remark}

\numberwithin{equation}{section}

\begin{document}


\title[RNA pseudoknot structures with arc-length $\ge 3$
and stack-lenght $\ge \sigma$]
      { RNA pseudoknot structures with arc-length $\ge 3$
and stack-lenght $\ge \sigma$}
\author{Emma Y. Jin and Christian M. Reidys$^{\,\star}$}
\address{Center for Combinatorics, LPMC-TJKLC 
           \\
         Nankai University  \\
         Tianjin 300071\\
         P.R.~China\\
         Phone: *86-22-2350-6800\\
         Fax:   *86-22-2350-9272}
\email{reidys@nankai.edu.cn}
\thanks{}
\keywords{RNA secondary structure, pseudoknot, enumeration, generating
function}
\date{November, 2007}
\begin{abstract}
In this paper we enumerate $k$-noncrossing RNA pseudoknot structures
with given minimum arc- and stack-length. That is, we study the
numbers  of RNA pseudoknot structures with arc-length $\ge 3$,
stack-length $\ge \sigma$ and in which there are at most $k-1$
mutually crossing bonds, denoted by ${\sf T}_{k,\sigma}^{[3]}(n)$.
In particular we prove that the numbers of $3$, $4$ and
$5$-noncrossing RNA structures with arc-length $\ge 3$ and
stack-length $\ge 2$ satisfy ${\sf T}_{3,2}^{[3]}(n)^{}\sim
K_3\,n^{-5} 2.5723^n$, ${\sf T}^{[3]}_{4,2}(n)\sim K_4\,n^{-\frac{21}{2}}
\,3.0306^n$, and ${\sf T}^{[3]}_{5,2}(n)\sim K_5\, n^{-18}\,3.4092^n$,
respectively, where $K_3,K_4,K_5$ are constants.
Our results are of importance for
prediction algorithms for RNA pseudoknot structures.
\end{abstract}
\maketitle
{{\small
}}


\section{Introduction}\label{S:intro}

An RNA structure is the helical configuration of an RNA sequence, described by
its primary sequence of nucleotides {\bf A}, {\bf G}, {\bf U} and {\bf C}
together with the Watson-Crick ({\bf A-U}, {\bf G-C}) and ({\bf U-G}) base
pairing rules specifying which pairs of nucleotides can potentially form bonds.
Subject to these single stranded RNA form helical structures.
The function of many RNA sequences is oftentimes tantamount to their
structures. Therefore it is of central importance to understand RNA
structure in the context of studying the function of biological RNA,
and in the design process of artificial RNA.
In this paper we enumerate $k$-noncrossing RNA structures with
arc-length $\ge 3$ and stack-length $\ge \sigma$, where $\sigma\ge
2$. The main idea is to consider a certain subset of
$k$-noncrossing core-structures, that is structures with minimum arc
length $2$, in which there exists {no} two arcs of the form
$(i,j),(i+1,j-1)$ and {no} arcs of the
form $(i,i+2)$ with isolated $i+1$. We prove a bijection between
this subset of core structures with multiplicities and $k$-noncrossing
RNA structures with arc-length $\ge 3$ and stack-length $\ge
\sigma$, where $\sigma\ge 2$. Subsequently, we derive several
functional equations of generating functions, based on which transfer
theorems imply our asymptotic formulas.
The paper is relevant for prediction algorithms
of pseudoknot RNA since it proves that the numbers  of
$k$-noncrossing RNA structures with arc-length $\ge 3$ and
stack-length $\ge \sigma$ exhibit small exponential growth rates.
The results suggest a novel strategy for RNA pseudoknot prediction.

\section{Diagrams, matchings and structres}

A diagram is labeled graph over the vertex set $[n]=\{1,\dots, n\}$ with
degree $\le 1$, represented by drawing its vertices $1,\dots,n$
in a horizontal line and its arcs $(i,j)$, where $i<j$, in the upper
halfplane. The vertices and arcs correspond to nucleotides and Watson-Crick
({\bf A-U}, {\bf G-C}) and ({\bf U-G}) base pairs, respectively.
\begin{figure}[ht]
\centerline{%
\epsfig{file=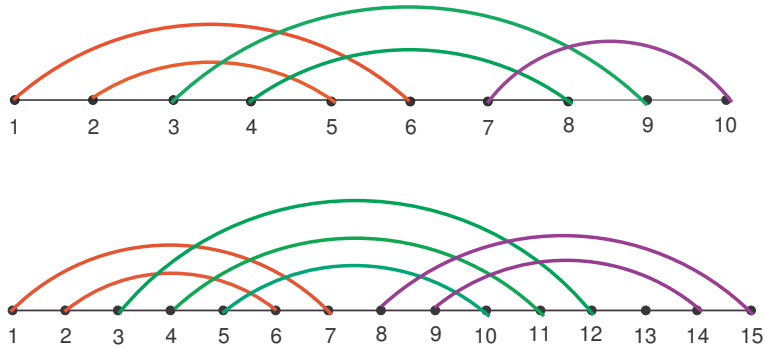,width=0.7\textwidth}\hskip15pt
 }
\caption{\small $k$-noncrossing diagrams. Top: $3$-noncrossing
diagram with arc-length $\ge 3$, $(2,5)$, $(7,10)$, the arc $(7,10)$
being isolated. Hence we have a $3$-noncrossing diagram $\lambda=3$,
$\sigma=1$ diagram without isolated vertices. Bottom:
$3$-noncrossing (no red/purple cross), $\lambda=4$, $\sigma=2$
diagram with isolated vertices $13$. } \label{F:lego0}
\end{figure}

We categorize diagrams according to the maximum number of mutually
crossing arcs, $k-1$, the  minimum arc-length, $\lambda$, and the
minimum stack-length, $\sigma$. Here the length of an arc $(i,j)$ is
$j-i$ and a stack of length $\sigma$ is a sequence of ``parallel''
arcs of the form
$((i,j),(i+1,j-1),\dots,(i+(\sigma-1),j-(\sigma-1)))$. In the
following we call a $k$-noncrossing diagram with arc-length $\ge 2$
and stack-length $\ge \sigma$ a $k$-noncrossing RNA structure. We
denote the set (number) of $k$-noncrossing RNA structures with
stack-size $\ge\sigma$ by $T_{k,\sigma}(n)$ (${\sf
T}_{k,\sigma}^{}(n)$) and refer to $k$-noncrossing  RNA structures
for $k\ge 3$ as pseudoknot RNA structures. A $k$-noncrossing
core-structure is a $k$-noncrossing RNA structures in which there
exists {\it no} two arcs of the form $(i,j),(i+1,j-1)$. We denote
the set (number) of core-structures having $h$ arcs by
${C}_{k}(n,h)$ (${\sf C}_{k}(n,h)$) and ${C}_{k}(n)$ (${\sf
C}_{k}(n)$) denotes the set (number) of core-structures. The set
(number) of RNA structures with arc-length $\ge 3$, is denoted by
$T_{k,\sigma}^{[3]}(n)$ (${\sf T}_{k,\sigma}^{[3]}(n)$). For $k=2$
we have RNA structures with no $2$ crossing arcs, i.e.~the
well-known RNA secondary structures, whose combinatorics was
pioneered by Waterman {\it et.al.}
\cite{Penner:93c,Waterman:79a,Waterman:78a,Waterman:94a,Waterman:80}.
RNA secondary structures are $T_{2,1}(n)$-structures. We denote by
$f_{k}(n,\ell)$ the number of $k$-noncrossing diagrams with
arbitrary arc-length and $\ell$ isolated points over $n$ vertices.
In Figure~\ref{F:4-pic} we display the various types of diagrams
involved.
\begin{figure}[ht]
\centerline{%
\epsfig{file=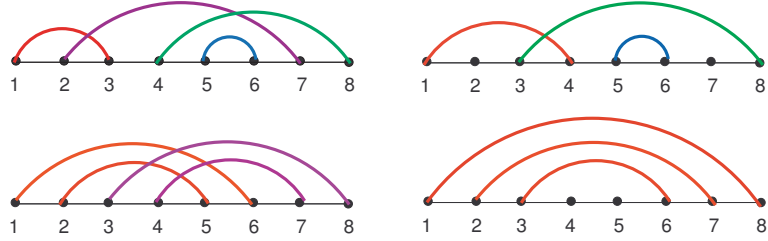,width=0.7\textwidth}\hskip15pt
 }
\caption{\small Basic diagram types: (a) perfect matching
($f_3(8,0)$), (b) partial matching with $1$-arc $(5,6)$ and isolated
points $2,7$ ($f_3(8,2)$), (c) structure with arc-length
$\ge 3$ and stack-length $\ge 2$ and no isolated points (${\sf
T}_{3,2}^{[3]}(8)$) and (d) structure with arc-length $\ge 3$,
stack-length $\ge 3$ and isolated points $4,5$ (${\sf T}_{2,3}^{[3]}(8)$).}
\label{F:4-pic}
\end{figure}

The following identities are due to Grabiner {\it et.al.}~\cite{Grabiner:93a}
\begin{eqnarray}\label{E:ww0}
\label{E:ww1}
\sum_{n\ge 0} f_{k}(n,0)\cdot\frac{x^{n}}{n!} & = &
\det[I_{i-j}(2x)-I_{i+j}(2x)]|_{i,j=1}^{k-1} \\
\label{E:ww2}
\sum_{n\ge 0}\left\{\sum_{\ell=0}^nf_{k}(n,\ell)\right\}\cdot\frac{x^{n}}{n!}
&= &
e^{x}\det[I_{i-j}(2x)-I_{i+j}(2x)]|_{i,j=1}^{k-1} \ ,
\end{eqnarray}
where $I_{r}(2x)=\sum_{j \ge 0}\frac{x^{2j+r}}{{j!(r+j)!}}$ denotes
the hyperbolic Bessel function of the first kind of order $r$.
Eq.~(\ref{E:ww1})
and (\ref{E:ww2}) allow ``in principle'' for explicit computation of
the numbers $f_k(n,\ell)$. In particular for $k=2$ and $k=3$ we have
the formulas
\begin{equation}\label{E:2-3}
f_2(n,\ell)  =  \binom{n}{\ell}\,C_{(n-\ell)/2}\quad
\text{\rm and}\quad  f_{3}(n,\ell)=
{n \choose \ell}\left[C_{\frac{n-\ell}{2}+2}C_{\frac{n-\ell}{2}}-
      C_{\frac{n-\ell}{2}+1}^{2}\right] \ ,
\end{equation}
where $C_m$ denotes the $m$-th Catalan number. $f_{3}(n,\ell)$ results from
a determinant formula enumerating pairs of nonintersecting Dyck-paths.
In view of
$
f_{k}(n,\ell) ={n \choose \ell} f_{k}(n-\ell,0)
$
everything can be reduced to perfect matchings, where we have the
following situation: there exists an asymptotic approximation of the
hyperbolic Bessel function due to \cite{Odlyzko:95a} and employing
the subtraction of singularities-principle \cite{Odlyzko:95a} one
can prove
\begin{equation}\label{E:f-k-imp}
\forall\, k\in\mathbb{N};\qquad  f_k(2n,0)\sim \varphi_k(n)
\left(\frac{1}{\rho_k}\right)^n\  ,
\end{equation}
where $\rho_k$ is the dominant real singularity of $\sum_{n\ge
0}f_k(2n,0)z^n$ and $\varphi_k(n)$ is a polynomial over $n$.
Via Hadamard's formula, $\rho_k$ can be expressed as
\begin{equation}\label{E:rho-k}
\rho_k=\lim_{n\to \infty}(f_k(2n,0))^{-\frac{1}{2n}} \ .
\end{equation}
Eq.~(\ref{E:f-k-imp}) allows for any given $k\ge 2$ to obtain
$\varphi_k(n)$, explicitly.

As for the generating function and asymptotics of $k$-noncrossing
RNA structures we have the following results from
\cite{Reidys:07pseu,Reidys:07asy1,Reidys:07asy2}. First the number
of $k$-noncrossing RNA structures with $(\frac{n-\ell}{2})$ arcs,
${\sf T}_{k,1}(n,\frac{n-\ell}{2})$, and the number of
$k$-noncrossing RNA structures, ${\sf T}_{k,1}(n)$, are given by
\begin{eqnarray}\label{E:ddd}
{\sf T}_{k,1}(n,\frac{n-\ell}{2})
& = & \sum_{b=0}^{\lfloor n/2\rfloor}(-1)^{b}{n-b \choose b}
f_{k}(n-2b,\ell) \\
\label{E:sum}
{\sf T}_{k,1}(n)
& = & \sum_{b=0}^{\lfloor n/2\rfloor}(-1)^{b}{n-b \choose b}
\left\{\sum_{\ell=0}^{n-2b}f_{k}(n-2b,\ell)\right\} \ ,
\end{eqnarray}
where $\{\sum_{\ell=0}^{n-2b}f_{k}(n-2b,\ell)\}$ is given via
eq.~{\rm (\ref{E:ww2})}.
Secondly we have
\begin{eqnarray}
\label{E:konk3} {\sf T}_{3,1}(n) & \sim & \frac{10.4724\cdot
4!}{n(n-1)\dots(n-4)}\,
\left(\frac{5+\sqrt{21}}{2}\right)^n \\
{\sf T}_{3,1}^{[3]}(n) & \sim &
\frac{6.11170\cdot 4!}{n(n-1)\dots(n-4)} \ 4.54920^n \ .
\end{eqnarray}

The particular class of $k$-noncrossing core-structures, i.e.~structures
in which there exists {no} two arcs of the form $(i,j),(i+1,j-1)$ will play
a central role in the following enumerations:
\begin{theorem}\cite{Reidys:07lego}\,{\bf (Core-structures)}\label{T:core}
Suppose $k\in\mathbb{N}$, $k\ge 2$,  let $x$ be an indeterminant,
$\rho_k$ the dominant, positive real singularity of
$\sum_{n\ge 0}f_k(2n,0)z^n$ (eq.~(\ref{E:rho-k})) and
$u_1(x)=\frac{1}{1+x^2}$. Then the numbers of $k$-noncrossing
core-structures over $[n]$, ${\sf C}_k(n)$ are given by
\begin{equation}\label{E:inversion}
{\sf C}_{k}^{ }(n,h)=
\sum_{b=0}^{h-1}(-1)^{h-b-1}{h-1 \choose b}
{\sf T}^{ }_{k,1}(n-2h+2b+2,b+1) \ .
\end{equation}
Furthermore we have the functional equation
\begin{eqnarray}\label{E:C1}
\sum_{n \ge 0}{\sf C}_{k}^{}(n)\ x^n & = &
\frac{1}{u_1x^2-x+1}\sum_{n \ge
0}f_{k}(2n,0)\left(\frac{\sqrt{u_1}x}{u_1x^2-x+1}\right)^{2n}
\end{eqnarray}
and
\begin{equation}\label{E:growth-free}
{\sf C}_{k}^{}(n)\sim \varphi_k(n)\ \left(\frac{1}{\kappa_{k}}\right)^n
\end{equation}
where $\kappa_{k}$ is a dominant singularity of
$\sum_{n \ge 0}{\sf C}_{k}^{}(n)$ and the minimal real solution
of the equation $\frac{\sqrt{u_1}\ x}{u_1x^2-x+1}  =  \rho_k$
and $\varphi_k(n)$ is a polynomial over $n$ derived from the
asymptotic expression of $f_k(2n,0)\sim \varphi_k(n)\left(\frac{1}
{\rho_k}\right)^n$ of eq.~(\ref{E:f-k-imp}).
\end{theorem}

The following functional identity \cite{Reidys:07asy1} relates
the bivariate generating function for ${\sf T}_{k,1}(n,h)$, the number of
RNA pseudoknot structures with $h$ arcs to the generating
function of $k$-noncrossing perfect matchings. It will be instrumental
for the proof of Theorem~\ref{T:arc-3} in Section~\ref{T:arc-3}.

\begin{lemma}\label{L:arc-2}
Let $k\in\mathbb{N}$, $k\ge 2$ and $z,u$ be indeterminants over
$\mathbb{C}$. Then we have
\begin{equation}\label{E:rr}
\sum_{n\ge 0} \sum_{h\le n/2} {\sf T}_{k,1}(n,h) \ u^{2h} z^n =
\frac{1}{u^2z^2-z+1}
\sum_{n\ge 0} f_k(2n,0) \left(\frac{uz}{u^2z^2-z+1}\right)^{2n} \ .
\end{equation}
In particular we have for $u=1$,
\begin{equation}\label{E:oha}
\sum_{n\ge 0} {\sf T}_{k,1}(n) \ z^{n} =
\frac{1}{z^2-z+1}\,
\sum_{n\ge 0} f_k(2n,0) \left(\frac{z}{z^2-z+1}\right)^{2n} \ .
\end{equation}
\end{lemma}

In view of Lemma~\ref{L:arc-2} it is of interest to deduce
relations between the coefficients from the equality of generating
functions. The class of theorems that deal with this deduction are
called transfer-theorems \cite{Flajolet:07a}. One key ingredient in
this framework is a specific domain in which the functions in
question are analytic, which is ``slightly'' bigger than their
respective radius of convergence. It is tailored for extracting the
coefficients via Cauchy's integral formula:
given two numbers $\phi,R$, where $R>1$ and $0<\phi<\frac{\pi}{2}$ and
$\rho\in\mathbb{R}$ the open domain $\Delta_\rho(\phi,R)$ is defined as
\begin{equation}
\Delta_\rho(\phi,R)=\{ z\mid \vert z\vert < R, z\neq \rho,\,
\vert {\rm Arg}(z-\rho)\vert >\phi\}
\end{equation}
A domain is a $\Delta_\rho$-domain if it is of the form
$\Delta_\rho(\phi,R)$ for some $R$ and $\phi$.
A function is $\Delta_\rho$-analytic if it is analytic in some
$\Delta_\rho$-domain.
We use the notation
\begin{equation}\label{E:genau}
\left(f(z)=O\left(g(z)\right) \
\text{\rm as $z\rightarrow \rho$}\right)\quad \Longleftrightarrow \quad
\left(f(z)/g(z) \ \text{\rm is bounded as $z\rightarrow \rho$}\right)
\end{equation}
and if we write $f(z)=O(g(z))$ it is implicitly assumed that $z$
tends to a (unique) singularity. $[z^n]\,f(z)$ denotes the
coefficient of $z^n$ in the power series expansion of $f(z)$ around
$0$.

\begin{theorem}\label{T:transfer}\cite{Flajolet:05}
Let $f(z),g(z)$ be a $\Delta_\rho$-analytic functions given by power
series $f(z)=\sum_{n\ge 0} a_nz^n$ and $g(z)=\sum_{n\ge 0}b_n z^n$.
Suppose $f(z)= O(g(z))$ for all $z\in\Delta_\rho$ and $b_n\sim
\varphi(n)(\rho^{-1})^n$, where $\varphi(n)$ is a polynomial over
$n$. Then
\begin{equation}
a_n = [z^n]\, f(z)  \sim  K \ [z^n]\, g(z)= K\,b_n\sim K\,
\varphi(n)(\rho^{-1})^n
\end{equation}
for some constant $K$.
\end{theorem}


\section{Exact Enumeration}\label{S:core-3}


Our first result, Theorem~\ref{T:core-3}, enumerates $k$-noncrossing RNA
structures with arc-length $\ge 3$ and stack-length $\ge \sigma$. The
structure of the formula is the exact analogue of the M\"obius inversion
of eq.~({\ref{E:inversion}}) \cite{Reidys:07lego}, which relates the
numbers of all structures and the numbers of core-structures:
$
{\sf T}_{k,\sigma}^{}(n,h)=\sum_{b=\sigma-1}^{h-1}
{b+(2-\sigma)(h-b)-1 \choose h-b-1} {\sf C}_k(n-2b,h-b)
$.
While the latter cannot be used in order to enumerate $k$-noncrossing
structures with arc-length $\ge 3$, see Figure~\ref{F:lego4},
the set
\begin{figure}[ht]
\centerline{%
\epsfig{file=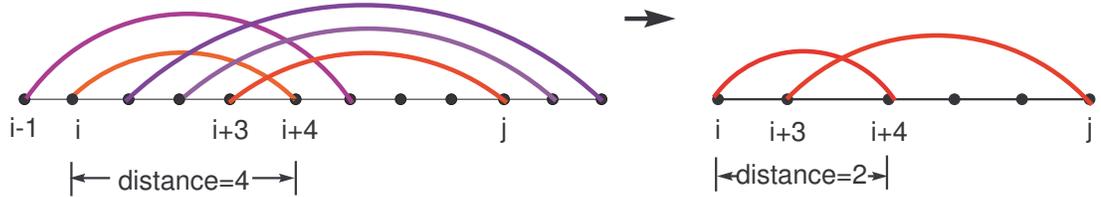,width=1\textwidth}\hskip15pt
 }
\caption{\small Core-structures will in general have $2$-arcs: the
structure $\delta\in T_{3,2}(12)$ (lhs) is mapped into its core
$c(\delta)$ (rhs). Clearly $\delta$ has arc-length $\ge 4$ and as a
consequence of the collapse of the stack
$((i+1,j+3),(i+2,j+2),(i+3,j))$ (the purple arcs are being removed)
into the arc $(i+3,j)$, $c(\delta)$ contains the arc $(i,i+4)$, which is,
after relabeling, a $2$-arc.
}
 \label{F:lego4}
\end{figure}
\begin{equation}\label{E:*}
C^*_k(n,h)=\{\delta\mid \delta\in C_k(n,h);\,  \not\exists\
(i,i+2);\ \text{\rm $i+1$ is an isolated vertex.}\,\}
\end{equation}
turns out to be the key.

\begin{theorem}\label{T:core-3}
Suppose we have $k,h,\sigma\in \mathbb{N}$, $k\ge 2$, $h\le n/2$ and
$\sigma\ge 2$. Then the numbers of $k$-noncrossing RNA structures with
arc-length $\ge 3$ and stack-length $\ge \sigma$ having $h$ arcs is given
by
\begin{equation}\label{E:relation1}
{\sf T}_{k,\sigma}^{[3]}(n,h)=\sum_{b=\sigma-1}^{h-1}
{b+(2-\sigma)(h-b)-1 \choose h-b-1} {\sf C}_k^*(n-2b,h-b)
\end{equation}
where ${\sf C}_{k}^{*}(n,h)$ satisfies
\begin{equation}\label{E:inv-3}
{\sf C}_{k}^{*}(n,h)=\sum_{j=0}^{h}(-1)^{j}
\binom{n-2j}{j} {\sf C}_k(n-3j,h-j)
\end{equation}
and ${\sf C}_k(n',h')$ is given by Theorem~\ref{T:core}.
\end{theorem}

\begin{center}
\begin{tabular}{c|cccccccccccccccccc}
\hline $n$ & \small$1$ & \small $2$ & \small $3$ & \small $4$
&\small $5$ & \small $6$ & \small $7$ & \small $8 $
& \small$9$ & \small$10$ & \small$11$ & \small$12$ & \small$13$ & \small$14$
& \small$15$ & \small $16$ & \small $17$ & \small $18$\\
\hline ${\sf T}_{3,2}^{[3]}(n)$ & \small$1$ & \small $1$ & \small$1$
& \small $1$ & \small $1$  & \small $2$ & \small$4$ & \small$9$ &
\small $19$ & \small $40$ & \small $82$ & \small$167$ & \small$334$
& \small$682$
& \small$1398$ & \small $2917$ & \small $6142$ & \small $13025$\\
\hline \small ${\sf T}_{3,3}^{[3]}(n)$ & \small $1$ &\small $1$ &
\small$1$ & \small $1$ & \small $1$ & \small$1$ & \small$1$&  \small
$2$ & \small$4$ & \small$8$ & \small$14$ & \small$24$ &\small $40$ &
\small$68$ & \small$118$
& \small $209$ & \small $371$ & \small $654$\\
\hline
\end{tabular}
\end{center}

\begin{proof}
We consider $C^*_k(n,h)$ (eq.~(\ref{E:*}))
and call an arc $(i,i+2)$ with isolated $i+1$ a {\it bad arc}. It is
straightforward to show that there are $\binom{n-2j}{j}$ ways to
select $j$ bad arcs over $[n]$. Since removing a bad arc by
construction removes $3$ vertices we observe that the number of
configurations of at least $j$ bad arcs is given by $\binom{n-2j}{j}
{\sf C}_k(n-3j,h-j)$. Via the inclusion-exclusion principle we
accordingly arrive at
\begin{equation}
{\sf C}_{k}^{*}(n,h)=\sum_{j=0}^{h}(-1)^{j}
\binom{n-2j}{j} {\sf C}_k(n-3j,h-j) \ .
\end{equation}
We next observe that there exists a mapping from $k$-noncrossing
structures with $h$ arcs  with arc-length $\ge 3$ and stack-length
$\sigma\ge 2$ over $[n]$ into $\dot\bigcup_{\sigma-1\le b\le h-1}C^*_k
(n-2b,h-b)$:
\begin{equation}
c\colon T^{[3]}_{k,\sigma}(n,h)\rightarrow
\dot\bigcup_{0\le b\le h-1}{C}^*_k(n-2b,h-b), \quad
\delta\mapsto c(\delta)
\end{equation}
which is obtained in two steps: first induce $c(\delta)$ by
mapping arcs and isolated vertices as follows:
\begin{equation}
\forall \,\ell\ge \sigma-1;\quad
((i-\ell,j+\ell),\dots,(i,j)) \mapsto (i,j) \ \quad \text{\rm and} \quad
j \mapsto j \quad \text{\rm if $j$ is an isolated vertex}
\end{equation}
and secondly relabel the resulting diagram from left to right in increasing
order, see Figure~\ref{F:lego2-x}.\\
\begin{figure}[ht]
\centerline{%
\epsfig{file=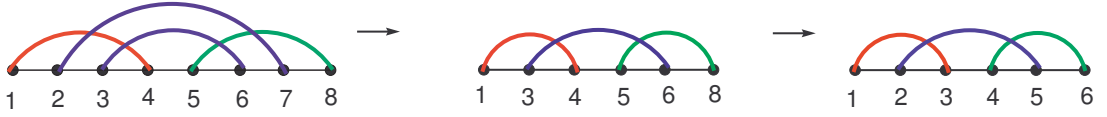,width=1\textwidth}\hskip15pt
 }
\caption{\small The mapping $c\colon
T_{k,\sigma}^{[3]}(n,h)\longrightarrow \dot\bigcup_{0\le b\le
h-1}{C}_k^{*}(n-2b,h-b)$ is obtained in two steps: first contraction
of the stacks and secondly relabeling of the resulting diagram. }
 \label{F:lego2-x}
\end{figure}
{\it Claim $1$.} $c\colon T^{[3]}_{k,\sigma}(n,h)\longrightarrow
\dot\bigcup_{\sigma-1\le b\le h-1}{C}^*_k(n-2b,h-b)$ is well-defined
and surjective.\\
By construction, $c$ does not change the crossing number.
Since $T^{[3]}_{k,\sigma}(n)$ contains only arcs of length $\ge 3$
we derive $c(T_{k,\sigma}^{[3]}(n))\subset C^*_k(n-2b,h-b)$.
Therefore $c$ is well-defined.
It remains to show that $c$ is surjective. For this purpose let $\delta\in
C^*_k(n-2b,h-b)$ and set $a=b-(\sigma-1)(h-b)$.
We proceed constructing a $k$-noncrossing structure $\tilde{\delta}$ in
three steps: \\
{\it Step $1$.} replace each label $i$ by $r_i$, where $r_i\le r_s$ if
and only if $i\le s$.\\
{\it Step $2$.} replace the leftmost arc $(r_p,r_q)$ by the sequence of arcs
\begin{equation}
\left((\tau_p-([\sigma-1]+a),\tau_q+([\sigma-1]+a)),\dots,(\tau_p,\tau_q)
\right)
\end{equation}
replace any other arc $(r_p,r_q)$ by the sequence
\begin{equation}
\left((\tau_p-[\sigma-1],\tau_q+[\sigma-1]),\dots,(\tau_p,\tau_q)\right)
\end{equation}
and each isolated vertex $r_s$ by $\tau_s$.\\
{\it Step $3$.} Set for $x,y\in\mathbb{Z}$, $\tau_b+y \le \tau_c+ x$ if and
only if ($b<c$) or ($b=c$ and $y\le x$). By construction, $\le$ is a linear
order over
$$
n-2b+2(h-b)\,(\sigma-1)+2a=
n-2b+2(h-b)\,(\sigma-1)+2(b-(\sigma-1)(h-b))=n
$$
elements, which we then label from $1$ to $n$ (left to right) in
increasing order.
It is straightforward to verify that $c(\tilde{\delta})=\delta$
holds. It remains to show that $\tilde{\delta}\in
T^{[3]}_{k,\sigma}(n)$. Suppose {\it a contrario} $\tilde{\delta}$
contains an arc $(i,i+2)$. Since $\sigma\ge 2$ we can then conclude
that $i+1$ is necessarily isolated. The arc $(i,i+2)$ is mapped by
$c$ into $(j,j+2)$ with isolated point $j+1$,
which is impossible by definition of $C^*_k(n',h')$ and Claim $1$ follows.\\
Labeling the $h$ arcs of $\delta\in {\sf T}^{[3]}_{k,\sigma}(n,h)$ from
left to right and keeping track of multiplicities gives rise to the
map
\begin{equation}\label{E:core}
f_{k,\sigma}^{} \colon T^{[3]}_{k,\sigma}(n,h) \rightarrow
\dot\bigcup_{0\le b\le h-1} \left[C_k^{*}(n-2b,h-b) \times
\left\{(a_{j})_{1\le j\le h-b}\mid\sum_{j=1}^{h-b}a_{j}=b, \
a_{j}\ge \sigma-1 \right\}\right],
\end{equation}
given by $f_{k,\sigma}^{ }(\delta)=(c(\delta),(a_{j})_{1\le j\le
h-b})$. We can conclude that $f_{k,\sigma}^{ }$ is well-defined and
a bijection. We proceed computing the multiplicities of the
resulting core-structures \cite{Reidys:07lego}:
\begin{equation}\label{E:core-1}
\vert\{(a_j)_{1\le j\le b}\mid \sum_{j=1}^{h-b}a_j=b; \ a_j\ge
\sigma -1\}\vert = {b+(2-\sigma)(h-b)-1 \choose h-b-1} \ .
\end{equation}
Eq.~(\ref{E:core-1}) and eq.~(\ref{E:core}) imply
\begin{equation}
{\sf T}^{[3]}_{k,\sigma}(n,h)=\sum_{b=\sigma-1}^{h-1}
{b+(2-\sigma)(h-b)-1 \choose h-b-1}{\sf C}^{*}_{k}(n-2b,h-b) \ ,
\end{equation}
and the theorem follows.
\end{proof}

We proceed by proving a functional identity between the bivariate
generating functions of ${\sf T}^{[3]}_{k,\sigma}(n,h)$ and ${\sf
C}^{*}_{k}(n,h)$. This identity is based on Theorem~\ref{T:core-3} and
crucial for proving Theorem~\ref{T:arc-3} in Section~\ref{S:arc-3}.
Its proof is analogous to Lemma~3 in \cite{Reidys:07lego}.

\begin{lemma}{\bf }\label{L:www-3}
Let $k,\sigma\in \mathbb{N}$, $k\ge 2$ and let $u,x$ be indeterminants.
Suppose we have
\begin{equation}\label{E:uni1}
\forall\, h\ge 1;\quad {\sf A}_{k,\sigma}^{}(n,h)=\sum_{b=\sigma-1}^{h-1}
{b+(2-\sigma)(h-b)-1 \choose h-b-1} {\sf B}_k(n-2b,h-b)\quad \text{\rm
and} \quad {\sf A}_{k,\sigma}^{}(n,0)=1 \ .
\end{equation}
Then we have the functional relation
\begin{eqnarray}\label{E:universal}
\sum_{n \ge 0}\sum_{h\le \frac{n}{2}}{\sf A}^{ }_{k,
\sigma}(n,h)u^hx^n & = &
\sum_{n \ge 0}\sum_{h \le \frac{n}{2}}{\sf B}^{ }_{k}
(n,h)\left(\frac{u\cdot (ux^2)^{
\sigma-1}}{1-ux^2}\right)^hx^n+\frac{x}{1-x}\ .
\end{eqnarray}
\end{lemma}
\begin{proof}
We set $\sum_{n \ge 0}\sum_{h\le\frac{n}{2}}{\sf B}_{k}^{ }
(n,h)u^hx^n=\sum_{h\ge 0}\varphi_{h}^{ }(x)u^h$ and compute in view
of eq.~(\ref{E:uni1})
\begin{equation}
\sum_{n \ge 0}\sum_{h\le \frac{n}{2}}
{\sf A}_{k,\sigma}^{ }(n,h)u^hx^n
 =
\sum_{n \ge 0}\sum_{h \le \frac{n}{2}}\sum_{b\le h-1}{\sf
B}_{k}^{ }(n-2b,h-b){b+(2-\sigma)(h-b)-1 \choose
h-b-1}u^h x^n + \sum_{i\ge 1}x^i
\end{equation}
where the term $\sum_{i\ge 1}x^i=\frac{x}{1-x}$ comes from the fact that
for $h=0$ the binomial
$$
{b+(2-\sigma)(h-b)-1 \choose h-b-1}
$$
is zero, while for any $i\ge 1$ the lhs counts ${\sf
A}_{k,\sigma}(n,0)=1$. We proceed by computing
\begin{eqnarray*}
&= & \sum_{h \ge 0}\sum_{b\le h-1}\sum_{n \ge 2h}{\sf B}_{k}^{ }
(n-2b,h-b)x^{n-2b}{b+(2-\sigma)(h-b)-1
\choose h-b-1}u^h x^{2b}+ \frac{x}{1-x}\\
&=& \sum_{b \ge 0}\sum_{b\le
h}\varphi_{h-b}^{ }(x){b+(2-\sigma)(h-b)-1 \choose h-b-1}u^h
x^{2b} + \frac{x}{1-x} \ .
\end{eqnarray*}
Setting $m=h-b$ and subsequently interchanging the summation indices we
arrive at
\begin{eqnarray*}
\sum_{n \ge 0}\sum_{h\le \frac{n}{2}}
{\sf A}_{k,\sigma}^{ }(n,h)u^hx^n
&= & \sum_{b \ge 0}\sum_{1\le m}\varphi_{m}^{ }(x){
b+(2-\sigma)m-1\choose m-1}u^m(ux^2)^b + \frac{x}{1-x}\\
&=& \sum_{m \ge
0}\varphi_{m}^{ }(x)\left(\frac{u\cdot(ux^2)^{\sigma-1}}
{1-ux^2}\right)^m + \frac{x}{1-x}\\
&=& \sum_{n \ge 0}\sum_{h \le \frac{n}{2}}{\sf B}_{k}^{ }
(n,h)\left(\frac{u\cdot
(ux^2)^{\sigma-1}}{1-ux^2}\right)^hx^n + \frac{x}{1-x} \ ,
\end{eqnarray*}
whence Lemma~\ref{L:www-3}.
\end{proof}

According to Lemma~\ref{L:www-3} and eq.~(\ref{E:relation1}) we have
\begin{eqnarray}\label{E:h1}
{\sf T}_{k,\sigma}(n,h) & = & \sum_{b=\sigma-1}^{h-1}
{b+(2-\sigma)(h-b)-1 \choose h-b-1} {\sf C}_k(n-2b,h-b)\\
\label{E:h2}
{\sf T}_{k,\sigma}^{[3]}(n,h) & = & \sum_{b=\sigma-1}^{h-1}
{b+(2-\sigma)(h-b)-1 \choose h-b-1} {\sf C}_k^*(n-2b,h-b) \ .
\end{eqnarray}
and Lemma~\ref{L:www-3} implies the following two
functional identities, which are instrumental for the proof
of Theorem~\ref{T:arc-3} in Section~\ref{S:arc-3}.
\begin{eqnarray}\label{E:a}
\sum_{n \ge 0}\sum_{h\le \frac{n}{2}}{\sf T}^{ }_{k,
\sigma}(n,h)u^hx^n & = &
\sum_{n \ge 0}\sum_{h \le \frac{n}{2}}{\sf
C}^{ }_{k}(n,h)\left(\frac{u\cdot (ux^2)^{
\sigma-1}}{1-ux^2}\right)^hx^n+\frac{x}{1-x}\\
\label{E:d}
\sum_{n \ge 0} {\sf T}^{[3] }_{k,\sigma}(n)x^n & = &
\sum_{n \ge 0}\sum_{h \le \frac{n}{2}}{\sf
C}^{*}_{k}(n,h)\left(\frac{ (x^2)^{\sigma-1}}{1-x^2}\right)^hx^n
+\frac{x}{1-x} \ .
\end{eqnarray}


\section{Asymptotic Enumeration}\label{S:arc-3}


In this Section we study the asymptotics of
$k$-noncrossing RNA pseudoknot structures with arc-length $\ge 3$
and minimum stack length $\sigma$.
We are particulary interested in deriving simple formulas, that can be
used assessing the complexity of prediction algorithms for $k$-noncrossing
RNA structures. In order to state Theorem~\ref{T:arc-3}
below we introduce the following rational functions
\begin{eqnarray}
\label{E:w}
w_0(x)       & = & \frac{x^{2\sigma-2}}{1-x^2}\\
\label{E:z}
z_{0}(x)     & = & \frac{x}{1+w_{0}(x)x^3}        \\
\label{E:u}
u_0(x)       & = & \frac{w_0}{1+w_0(x)z_0(x)^2}       \\
\label{E:v}
v_0(x)       & = & \frac{1}{x}\frac{z_0(x)}{u_0(x)z_0(x)^2-z_0(x)+1} \ .
\end{eqnarray}

\begin{theorem}\label{T:arc-3}
Let $k,\sigma\in\mathbb{N}$, $k,\sigma\ge 2$, $x$ be an indeterminant
and $\rho_k$ the dominant, positive real singularity of $\sum_{n\ge
0}f_k(2n,0)z^n$ (eq.~(\ref{E:rho-k})).
Then ${\sf T}^{[3]}_{ k,\sigma}(n)$,
the number of RNA structures with arc-length $\ge 3$ and $\sigma\ge 2$
satisfies the following identity
\begin{equation}\label{E:functional-arc-3}
\sum_{n \ge 0}{\sf T}_{k,\sigma}^{[3]}(n)x^n=
v_0(x)\,\sum_{n \ge
0}f_k(2n,0)\left(\frac{\sqrt{u_{0}(x)}z_0(x)}{u_{0}(x)
z_{0}(x)^2-z_{0}(x)+1}\right)^{2n}+\frac{x}{1-x}-
\frac{z_0(x)^2}{1-z_0(x)}
\end{equation}
where $w_0(x),z_0(x),u_0(x),v_0(x)$ are given by
eq.~(\ref{E:w})--eq.~(\ref{E:v}).
Furthermore
\begin{equation}\label{E:growth-3}
{\sf T}_{k,\sigma}^{[3]}(n) \sim\varphi_{k}(n)\,
\left(\frac{1}{\gamma_{k,\sigma}^{[3]}}\right)^n
\end{equation}
holds, where $\gamma_{k,\sigma}^{[3]}$ is the positive real dominant
singularity of $\sum_{n \ge 0}{\sf T}_{k,\sigma}^{}(n)z^n$ and
minimal real solution of the equation
\begin{equation}
\frac{\sqrt{u_{0}(x)}z_0(x)}{u_{0}(x)z_{0}(x)^2-z_{0}(x)+1}=\rho_k
\end{equation}
and $\varphi_k(n)$ is a polynomial over $n$ derived from the
asymptotic expression of $f_k(2n,0)\sim \varphi_k(n)\left(\frac{1}
{\rho_k}\right)^n$ of eq.~(\ref{E:f-k-imp}).
\end{theorem}
Theorem~\ref{T:arc-3} implies the following growth rates for $3$-, $4$- and
$5$-noncrossing RNA structures with arc-length $\ge 3$ and stack-length
$\ge 2,3$:
\begin{eqnarray*}
&& (\gamma_{3,2}^{[3]})^{-1} = 2.5723 \qquad \qquad
(\gamma_{4,2}^{[3]})^{-1}=3.0306   \qquad \qquad
   (\gamma_{5,2}^{[3]})^{-1} = 3.4092    \\
&& (\gamma_{3,3}^{[3]})^{-1} = 2.0392 \qquad \qquad
(\gamma_{4,3}^{[3]})^{-1}=2.2663     \qquad  \qquad
(\gamma_{5,3}^{[3]})^{-1}=2.4442 \ .
\end{eqnarray*}
Furthermore we compare in the following table the exact subexponential
factors ${\sf T}_{k,3}^{[3]}(n)\,(\gamma_{k,2}^{[3]})^{n}$ computed
via Theorem~\ref{T:core-3} with the subexponential factors $\varphi_k(n)$
obtained from Theorem~\ref{T:transfer}:
\begin{center}
\begin{tabular}{|c|c|c|c|c|c|c|}
\hline
  \multicolumn{7}{|c|}{\bf The subexponential factor}\\
\hline $n$ & ${\sf T}_{3,3}^{[3]}(n)\,(\gamma_{3,3}^{[3]})^{n}$ &
$t_{3,3}^{[3]}(n)$ &
${\sf T}_{4,3}^{[3]}(n)\,(\gamma_{4,3}^{[3]})^{n}$ & $t_{4,3}^{[3]}(n)$
& ${\sf T}_{5,3}^{[3]}(n)\,(\gamma_{5,3}^{[3]})^{n}$ & $t_{5,3}^{[3]}(n)$\\
\hline \small$50$ & \small $4.93\times 10^{-5}$ & \small $3.97\times
10^{-6}$ & \small$4.59\times 10^{-7}$ & \small$2.13\times 10^{-10}$
& \small$1.53\times 10^{-8}$ & \small$7.81\times 10^{-9}$\\
\hline \small$60$ & \small $2.36\times 10^{-5}$ & \small$1.54\times
10^{-6}$
& \small$4.98\times 10^{-8}$ & \small$3.14\times 10^{-11}$
& \small$2.11 \times 10^{-9}$ & \small$2.93\times 10^{-10}$\\
\hline \small$70$ & \small $1.25\times 10^{-5}$ & \small$6.94\times
10^{-7}$
& \small$2.86\times 10^{-9}$ & \small$6.22\times 10^{-12}$
& \small$3.24\times 10^{-10}$ & \small$1.83\times 10^{-11}$\\
\hline \small$80$ & \small $7.07\times 10^{-6}$ & \small$3.49\times
10^{-7}$
& \small$1.26\times 10^{-10}$ & \small$1.53\times 10^{-12}$
& \small$1.71\times 10^{-12}$ & \small$1.65\times 10^{-12}$\\
\hline \small$90$ & \small $4.25\times 10^{-6}$ & \small$1.91\times
10^{-7}$
& \small$4.64\times 10^{-12}$ & \small$4.44\times 10^{-13}$
& \small$6.00\times 10^{-13}$ & \small$1.99\times 10^{-13}$\\
\hline \small$100$ & \small $2.68\times 10^{-6}$ & \small$1.12\times
10^{-7}$
& \small$1.47\times 10^{-13}$ & \small$1.47\times 10^{-13}$ & \small$2.98\times10^{-14}$
& \small$2.98\times 10^{-14}$\\
\hline
\end{tabular}
\end{center}
\begin{proof}
In the following we will use the notation $w_{0},u_{0},z_{0}$,
eq.~(\ref{E:w})--eq.~(\ref{E:v}),
for short without specifying the variable $x$. The first step consists
in deriving a functional equation relating the bivariate generating functions
of $C^*_k(n,h)$ and $C_k(n',h')$. For this purpose
we use eq.~(\ref{E:inv-3})
$
{\sf C}_{k}^{*}(n,h)=\sum_{b\le\frac{n}{2}}(-1)^b{n-2b \choose b}
{\sf C}_{k}^{}(n-3b,h-b)
$.\\
{\it Claim $1$.}
\begin{eqnarray}\label{E:dagger}
\sum_{n\ge 0}\sum_{h \le \frac{n}{2}}{\sf C}_{k}^{*}(n,h)w^hx^n &=&
\frac{1}{1+wx^3}\sum_{n \ge 0}\sum_{h
\le\frac{n}{2}}{\sf C}_{k}^{}(n,h)w^h(\frac{x}{1+wx^3})^n \ .
\end{eqnarray}
To prove Claim $1$ we compute
\begin{eqnarray*}
\sum_{n\ge 0}\sum_{h \le \frac{n}{2}}{\sf C}_{k}^{*}(n,h)w^hx^n &=&
\sum_{n \ge 0}\sum_{h \le \frac{n}{2}}\sum_{b\le h}(-1)^b{n-2b \choose b}
{\sf C}_{k}^{}(n-3b,h-b)w^hx^n \\
&=&\sum_{n \ge 0}\sum_{b\le\frac{n}{3}}\sum_{h \le \frac{n}{2}}
(-1)^b{n-2b \choose b}{\sf C}_{k}^{}(n-3b,h-b)w^hx^n\\
&=&\sum_{n \ge 0}\sum_{b \le \frac{n}{3}}(-1)^b{n-2b \choose b}
\sum_{h \le \frac{n}{2}}{\sf C}_{k}^{}(n-3b,h-b)w^hx^n \ .\\
\end{eqnarray*}
We rearrange the summation over $h$ and arrive at
$$
\sum_{n \ge 0}\sum_{b \le \frac{n}{3}}(-1)^b{n-2b \choose b}
\left[\sum_{h \le
\frac{n}{2}}{\sf C}_{k}^{}(n-3b,h-b)w^{h-b}\right]w^bx^n \ .
$$
Setting $\varphi_{n}(w)=\sum_{j\le\frac{n}{2}}
{\sf C}_{k}^{}(n,j)w^j$ this becomes
\begin{eqnarray*}
&=&\sum_{n \ge 0}\sum_{b \le\frac{n}{3}}(-1)^b{n-2b \choose
b}\varphi_{n-3b}(w)w^bx^n \\
&=&\sum_{b \ge 0}\frac{(-wx^3)^b}{b!}\sum_{n \ge 3b}
\frac{(n-2b)!}{(n-3b)!}\varphi_{n-3b}(w)x^{n-3b} \mbox{\quad where } \
m=n-3b  \\
&=&\sum_{b \ge 0}\frac{(-wx^3)^b}{b!}\sum_{m \ge
0}\frac{(m+b)!}{m!}\varphi_{m}(w)x^{m}\\
&=&\sum_{m \ge 0}\varphi_{m}(x)\frac{x^m}{m!} \sum_{b \ge
0}\frac{(-wx^3)^b}{b!}(m+b)!
\end{eqnarray*}
Laplace transformation of the series $\sum_{b \ge
0}\frac{(m+b)!}{b!}y^b$ yields
\begin{eqnarray*}
\sum_{b \ge 0}\frac{(m+b)!}{b!}y^b & = & \int_{0}^{\infty}\sum_{b
\ge 0}\frac{y^b}{b!}t^{m+b}e^{-t}dt\\
&=&\int_{0}^{\infty}t^m e^{-(1-y)t}dt\\
&=&\frac{1}{(1-y)^{m+1}}\int_{0}^{\infty} ((1-y)t)^m
e^{-(1-y)t}d((1-y)t)\\
&=&\frac{m!}{(1-y)^{m+1}} \ .
\end{eqnarray*}
Hence the bivariate generating function can be written as
\begin{eqnarray*}
\sum_{n \ge 0}\sum_{h \le\frac{n}{2}}C_{k}^{*}(n,h)w^hx^n & = &
\sum_{m \ge
0}\frac{\varphi_{m}(w)x^m}{m!}\frac{m!}{(1+wx^3)^{m+1}}\\
&=&\frac{1}{1+wx^3}\sum_{m \ge
0}\varphi_{m}(w)(\frac{x}{1+wx^3})^m\\
&=&\frac{1}{1+wx^3}\sum_{n \ge 0}\sum_{h \le\frac{n}{2}}{\sf
C}_{k}^{}(n,h)w^h(\frac{x}{1+wx^3})^n
\end{eqnarray*}
and the Claim $1$ follows.
According to eq.~(\ref{E:d}), we have
\begin{eqnarray}\label{E:aha0}
\sum_{n \ge 0}{\sf T}_{k,\sigma}^{[3]}(n)x^n & = &\sum_{n \ge 0}\sum_{h
\le\frac{n}{2}} {\sf C}_{k}^{*}(n,h)\left(\frac{ (x^2)^{\sigma-1}}{1-x^2}
\right)^hx^n
+\frac{x}{1-x} \ ,
\end{eqnarray}
and Claim $1$ provides, setting
\begin{equation}\label{E:w0}
w_{0}=\frac{ (x^2)^{\sigma-1}}{1-x^2} \ ,
\end{equation}
the
following interpretation of the rhs of eq.~(\ref{E:aha0}):
\begin{eqnarray}\label{E:aha}
\sum_{n \ge 0}\sum_{h \le\frac{n}{2}} {\sf C}_{k}^{*}(n,h)
\left(\frac{ (x^2)^{\sigma-1}}{1-x^2}\right)^hx^n
&=&\frac{1}{1+w_{0}x^3}\left[\sum_{n \ge 0}\sum_{h
\le\frac{n}{2}}{\sf
C}_{k}^{}(n,h)w_{0}^h(\frac{x}{1+w_{0}x^3})^n\right] \ .
\end{eqnarray}
According to eq.~(\ref{E:a}) and Lemma~\ref{L:arc-2} we have
\begin{eqnarray}\label{E:b1}
\sum_{n
\ge 0}\sum_{h \le\frac{n}{2}}{\sf T}_{k,1}(n,h)u^hz^n & = &
\left[\sum_{n \ge 0}\sum_{h \le \frac{n}{2}}{\sf C}_{k}^{}(n,h)(
\frac{u}{1-uz^2})^hz^n \right]+\frac{z}{1-z}\\
\label{E:b2}
\sum_{n \ge 0}\sum_{h \le\frac{n}{2}}{\sf T}_{k,1}(n,h)u^hz^n
&=&\frac{1}{uz^2-u+1}\sum_{n \ge
0}f_k(2n,0)\left(\frac{\sqrt{u}z}{uz^2-z+1}\right)^{2n}
\end{eqnarray}
and consequently
\begin{equation}\label{E:aha2}
\sum_{n \ge 0}\sum_{h \le \frac{n}{2}}{\sf C}_{k}^{}(n,h)(
\frac{u}{1-uz^2})^hz^n = \frac{1}{uz^2-z+1}\sum_{n \ge
0}f_k(2n,0)\left(\frac{\sqrt{u}z}{uz^2-z+1}\right)^{2n}-\frac{z}{1-z}
\end{equation}
holds. Suppose now
\begin{eqnarray}\label{E1}
z_0(x) & = &  \frac{x}{1+w_0x^3} \\
\label{E:2} u_0(x) & = &  \frac{w_0}{1+w_0z_0^2} \ .
\end{eqnarray}
Then we have the formal identity $\frac{u_0}{1-u_0z_0^2}=w_0$ and
 obtain substituting eq.~(\ref{E:aha}) into eq.~(\ref{E:aha0})
\begin{eqnarray*}
\sum_{n \ge 0}{\sf T}_{k,2}^{[3]}(n)x^n & =  &
\frac{1}{1+w_{0}x^3}\left[\sum_{n \ge 0}\sum_{h\le\frac{n}{2}} {\sf
C}_{k}^{}(n,h)w_{0}^h(\frac{x}{1+w_{0}x^3})^n \right] +\frac{x}{1-x} \ .
\end{eqnarray*}
Next we use eq.~(\ref{E:b1}) and the formal identity
$\frac{u_0}{1-u_0z_0^2}=w_0$ to arrive at:
\begin{eqnarray*}
\sum_{n \ge 0}{\sf T}_{k,2}^{[3]}(n)x^n & = & \frac{1}{1+w_{0}x^3} \
\left[ \sum_{n \ge 0}\sum_{h \le \frac{n}{2}}{\sf C}_{k}^{}(n,h)
(\frac{u_0}{1-u_0z_0^2})^hz_0^n \right]+\frac{x}{1-x} \ .
\end{eqnarray*}
Setting $t_0(x)=\frac{1}{1+w_0x^3}$, $t_1(x)=\frac{x}{1-x}$ and
$t_2(x)= \frac{z_0}{1-z_0}$ we derive via eq.~(\ref{E:aha2})
\begin{equation}\label{E:well}
\sum_{n \ge 0}{\sf T}_{k,2}^{[3]}(n)x^n  = t_0(x)\, \
\left[\frac{1}{u_0z_0^2-z_0+1}\sum_{n \ge
0}f_k(2n,0)\left(\frac{\sqrt{u_0}z_0}{u_0z_0^2-z_0+1}\right)^{2n}\right]
-t_0(x) t_2(x) + t_1(x) \ .
\end{equation}
We compute
\begin{eqnarray*}
t_0(x) & = & \frac{x-x^3}{x^{2\sigma+1}-x^2+1}  \\
t_1(x)-t_0t_2(x) &=&
\frac{f(x)}{(1-x-x^2+x^3+x^{2\sigma+1})(x-1)(x^2-1-x^{2\sigma+1})} \ ,
\end{eqnarray*}
where $f(x)$ is a polynomial of degree $4\sigma+2$. Therefore $t_0(x)$ and
$t_1(x)-t_0t_2(x)$ are analytic in $\{z\in\mathbb{C}\mid 0< \vert z\vert<
\frac{1}{2}\}$ and accordingly do not introduce any singularities of
$$
\frac{1}{u_0z_0^2-z_0+1}\sum_{n \ge
0}f_k(2n,0)\left(\frac{\sqrt{u_0}z_0}{u_0z_0^2-z_0+1}\right)^{2n}
$$
in the domain $\{z\in\mathbb{C}\mid 0< \vert z\vert< \frac{1}{2}\}$.
Furthermore we can conclude from eq.~(\ref{E:well}) that $x=0$ is a
removable singularity.
Let us denote
$
V(z)=\sum_{n \ge 0}f_k(2n,0)\left(\frac{\sqrt{u_0}z_0}
                                  {u_0z_0^2-z_0+1}\right)^{2n}
$.\\
{\it Claim $2$.} All dominant singularities of $\sum_{n \ge 0}
{\sf T}_{k,2}^{[3]}(n)z^n$ are singularities of $V(z)$. Furthermore
the unique, minimal, positive, real solution of
\begin{equation}
\frac{\sqrt{u_0}z_0}{u_0z_0^2-z_0+1} = \rho_k
\end{equation}
denoted by $\gamma_{k,\sigma}^{[3]}$ is a dominant singularity of
$\sum_{n\ge 0} {\sf T}_{k,\sigma}^{[3]}(n) z^n$. \\
Clearly, a dominant singularity of
$
\frac{1}{u_0z^2-z+1} V(z)
$
is either a singularity of $V(z)$ or $\frac{1}{u_0z^2-z+1}$. Suppose there
exists some singularity $\zeta\in\mathbb{C}$ which is a root of
$\frac{1}{u_0z^2-z+1}$. By construction $\zeta\neq 0$ and $\zeta$ is
necessarily a non-finite singularity of $V(z)$. If $\vert \zeta\vert \le
\gamma_{k,\sigma}^{[3]}$, then we arrive at the contradiction
$\vert V(\zeta)\vert>\vert V(\kappa_k)\vert$ since $V(\zeta)$
is not finite and
$$
V(\kappa_k)=\sum_{n\ge 0}f_k(2n,0)\rho_k^{2n}<\infty \ .
$$
Therefore all dominant singularities of
$\sum_{n\ge 0} {\sf T}_{k,\sigma}^{[3]}(n) z^n$
are singularities of $V(z)$. According to  Pringsheim's
Theorem~\cite{Titmarsh:39}, $\sum_{n\ge 0} {\sf T}_{k,\sigma}^{[3]}(n) z^n$
has a dominant positive real singularity which by construction equals
$\gamma_{k,\sigma}^{[3]}$ being the minimal positive real solution of
\begin{equation}
\frac{\sqrt{\frac{(x^2)^{\sigma-1}}{(x^2)^{\sigma}-x^2+1}}\ x}
{\left(\frac{(x^2)^{\sigma-1}}{(x^2)^{\sigma}-x^2+1}\right)\,
x^2-x+1}=\rho_k
\end{equation}
and the Claim $2$ follows.\\
Claim $2$ immediately implies that the inverse of ${\gamma_{k,\sigma}^{[3]}}$
equals the exponential growth-rate. According to \cite{Wang:07} the power
series $\sum_{n\ge 0}f_k(2n,0)z^n$ has an analytic continuation in a
$\Delta_{\rho_k}$-domain and we have $[z^n]V(z)\sim K \varphi_k(n)
(\rho^{-1})^n$, where $\varphi_k(n)$ is given by eq.~(\ref{E:f-k-imp}).
We can therefore employ Theorem~\ref{T:transfer}, which
allows us to transfer the subexponential factors from the asymptotic
expressions for $f_k(2n,0)$ to ${\sf T}_{k,\sigma}(n)$ and
eq.~(\ref{E:growth-3}) follows. This completes the proof of
Theorem~\ref{T:arc-3}.
\end{proof}


{\bf Acknowledgments.}
We are grateful to Prof.~W.Y.C.~Chen for stimulating discussions and helpful
comments.
This work was supported by the 973 Project, the PCSIRT Project of the
Ministry of Education, the Ministry of Science and Technology, and
the National Science Foundation of China.

\bibliographystyle{plain}



\end{document}